\documentclass[11pt]{article}
\usepackage{amsmath,amssymb}
\usepackage{color}
\newtheorem{thm}{Theorem}

\newtheorem{prop}{Proposition}

\newtheorem{rem}{Remark}

\newtheorem{ex}{Example}
\newenvironment{proof}{\noindent{\bf Proof:}\enspace}{{\vrule height6pt width4pt
depth3pt}\par\vspace{\medskipamount}\bigbreak}

\newcommand{\st}{s.t. & \begin{tabular}[t]{rcl}}

\newcommand{\hsig}{{\hat \sigma}}

\newcommand{\hmu}{{\hat \mu}}
\newcommand{\bmu}{{\bar \mu}}
\newcommand{\bsig}{{\bar\sigma}}

\usepackage{verbatim}
\usepackage{amsfonts}
\usepackage[pdftex]{graphicx}
\newcommand{\Arob}{{{\cal A}_{\mathrm{rob}}}(w_0)}
\newcommand{\uopt}{u_{opt}}
\newcommand{\Heps}{H-\epsilon}
\newcommand{\drifts}{\mathcal{U}_{\epsilon}}
\begin{document}

\title {  The Robust Merton Problem of an Ambiguity Averse Investor}
\author{Sara Biagini and  Mustafa \c{C}. P{\i}nar
\thanks{sara.biagini@ec.unipi.it, Department of Economics and Management, University of Pisa, 56100 Pisa, Italy; mustafap@bilkent.edu.tr. Department of Industrial Engineering
Bilkent University  06800 Bilkent, Ankara, Turkey.}} \maketitle

\begin{abstract}
We derive a closed form  portfolio optimization rule for an  investor who is  diffident about mean return  and volatility estimates, and has a CRRA utility. The  novelty is that confidence is here represented using ellipsoidal uncertainty sets for the drift, given a volatility realization.   This specification affords a  simple and concise analysis, as the optimal portfolio allocation policy is shaped by a rescaled market Sharpe ratio, computed under the worst case volatility. The result is  based on a  max-min Hamilton-Jacobi-Bellman-Isaacs PDE,  which extends the classical Merton problem  and reverts to it for an ambiguity-neutral investor.



\noindent{\bf Acknowledgements:} We  sincerely  thank  an anonymous Associate Editor   for his comments and suggestions, Fausto Gozzi  and Paolo Guasoni.  Part of this research has been conducted  while  Sara  Biagini was visiting   the London School of Economics and Political Sciences, and  special thanks go to  Constantinos Kardaras   for a number of precious  conversations on the topic.
\end{abstract}

\section{Introduction  }
Traditionally, financial modelling heavily relies on the choice of an underlying probability measure $P$, which is chosen to incorporate the statistical and stochastic nature of market price movements. As early back as the works  of Bachelier, Samuelson  and Black, Scholes and Merton,   the underlying risk factors---such as stock prices or interest rates---have been  modeled as Markovian diffusions (with possible jumps) under $P$. However,  as has become quite agreed upon, the complexity of the global economic and financial dynamics render impossible the precise identification of  the probability law of the evolution of the risk factors. Unavoidably, financial modelling is inherently subject to model uncertainty, which also appears under the appellation of Knightian uncertainty. \\
\indent In the presence of model uncertainty, one may admit various degrees of severity. One may deal with model misspecification only at the  level of the equivalence class of $P$, or go beyond and take into account a family of non dominated models.      The  core issue  of  portfolio  optimization has  been  widely investigated over the  last twenty years in the multiple priors context.  The investor has a pessimistic view of the odds, and takes a max-min  (also known as \emph{robust}) approach to the problem, first minimizing a utility functional over the priors and  afterwards  maximizing over the investment strategies.    We are aware of only a few   results  in the non dominated case, notably Hern\`{a}ndez-Hern\`{a}ndez and Schied \cite{hs} and the recent preprints by Nutz \cite{nutz} in full generality but discrete time,  and  Lin and Riedel \cite{riedel} in a  diffusion context.  On the contrary, in the dominated priors case  there is a rich  literature. We content ourselves with citing Chen and Epstein \cite{chenep}, Garlappi et alii \cite{gar}, Maenhout \cite{maenhout},  F\"{o}llmer et alii \cite{fsw} for a comprehensive review and references, and, more recently, the work by  Owari \cite{owari}.  \\
\indent  In such an active environment,  the present note offers a    resolution of the  robust non-dominated  Merton problem,     which is both simple and mathematically rigorous.  The main novelties  of the present contribution  lie  in the form of the uncertainty set   and  in the    accommodation for market incompleteness.  We assume that  the asset prices process  is  an $N$-dimensional diffusion,  and the   driving Wiener process is $d$-dimensional with $d\geq N$.   The investor is diffident about the  constant drift and volatility estimates  $\hat \mu$  and $\hsig$. Thus, she considers as plausible all the variance-covariance matrices lying in a given  compact set and,  for a given realization of $\sigma$, she considers  all the  drifts which take values in a   ellipsoid centered at $\hat \mu$:
$$
U_\epsilon(\sigma) = \{ u \in \mathbb{R}^N\mid  {(  u  - \hat \mu)' {(\Sigma)^{-1} }(  u  - \hat \mu)}\leq \epsilon^2 \},
$$
in which $\epsilon>0$ is the\emph{ radius of ambiguity} and $\Sigma =\sigma \sigma'$ is the variance-covariance matrix.  \\
 \indent The merits of an ellipsoidal representation for the ambiguous drifts has been amply demonstrated and discussed in \cite{gar},   \cite{goldiyen} for the  robust mean-variance optimization. The problem of worst-case
(max-min)
robust
portfolio choice is a well-studied problem (see e.g.,  \cite{delageye,demig1,fabbook,fab2,fabr,gar,goldiyen,pinatut,pop} for robust portfolio optimization in single period problems) under different representations of
ambiguity. Intuitively, the non-linear but simple geometry of ellipsoids
offers robustness that avoids a worst case which is a corner solution.  This  is the case  in a
polyhedral hyper-rectangle or box representation, as in  Lin and Riedel where the drift (as well as volatility)  is allowed to vary in a box $[\underline{\mu}, \overline{\mu}]$.   In the  dominated setup, the assumption of $k$-ignorance  in \cite{chenep} also amounts to a box representation for the drift. \\
\indent   At the same time, the choice of $U_\epsilon(\sigma)$  preserves
tractability. Citing Fabozzi et al. \cite{fabbook}:
``The coefficient realizations are assumed to be close to the forecasts, but
they may deviate. They are more likely to deviate from their (instantaneous) means if
their variability (measured by their standard deviation) is higher, so
deviations from the mean are scaled by the inverse of the covariance
matrix of the uncertain coefficients. The parameter $\epsilon$ corresponds to the
overall amount of scaled deviations of the realized returns from the
forecasts against which the investor would like to be protected.'' \\
\indent  Another appealing feature of taking  model uncertainty into account is that it  offers a theoretical solution to the equity premium puzzle. As noted by Mehra and Prescott \cite{mp}, the high levels of  historical equity premium  and  the simultaneous  moderate equity demand   seem  to be implied by unreasonable levels of risk aversion.  Their conclusion was skeptical on the ability of a frictionless Arrow-Debreu economy to account for such empirical evidence. However, the works by   Abel \cite{a} and Cecchetti, Lam and Mark \cite{clm}  addressed  the equity premium puzzle by relaxing the hypothesis that the investor perfectly knows the probability law.  The key point is that, in the multiple priors setup, the optimal equity demand depends on \emph{two aversion components}: risk and ambiguity aversion.    In accordance to these results, and the subsequent   \cite{chenep}, \cite{maenhout} and \cite{riedel},  we find that robustness of decisions \emph{lowers}   the  optimal demand of equity since the ambiguity and risk averse investor effectively behaves like a risk averse investor with an increased risk aversion coefficient. Precisely, in a CRRA utility case with relative risk aversion parameter $R$,  the optimal relative portfolio is given by
$$ \pi_\epsilon = \frac{(\overline{H}- \epsilon)^+ }{R \overline{H}}  \overline{\Sigma} ^{-1} (\hat \mu - r {\bf 1}) $$
 in which  $ \overline{\Sigma}$ is the worst case variance-covariance matrix, and $\overline{H}$ is the Sharpe ratio computed under $\overline{\Sigma}$ (see Proposition 1, 2 and Section 4).
So, when the ambiguity radius  is too high, namely it exceeds the worst case Sharpe ratio,   the investor refrains  from investing in the risky assets and puts all the money in the safe asset. The opposite case is when there is no uncertainty in the drift, that is  $\epsilon =0$, and  no uncertainty on volatility as well: then the optimal solution reverts to the Merton relative portfolio, $\pi_M: =    \frac{1}{R}  {\Sigma} ^{-1} (\hat \mu - r {\bf 1}) $.\\
  \indent The paper is organized as follows.  Section 2  contains the (diffusion) model specifications in  the non dominated case, together with the general version of the Martingale Principle needed here.  However, to derive an abstract max-min PDE from the Martingale Principle, some conditions on the volatility structure must be imposed in the fully incomplete market model, as in \cite{hs} for the case of a traded asset with coefficients  depending on an underlying, non traded, asset.  The focus here is  on the complete market case, to wit  the volatility is a square matrix. This allows for a simpler,  yet effective, analysis.  In fact,   the HJB-Isaacs  PDE formulation   shows that the investor  is observationally equivalent to one who has distorted, worst case, beliefs on the parameters.  In Section 3, a representative investor with   CCRA  utility is considered,  and  we assume  that   ambiguity is present only in the drift, i.e.  the priors are all equivalent. This is an interesting case per se, since the drift is subject to imprecision in estimations to a  much greater extent than volatility.  There, we solve  and provide the explicit solutions to the robust problems    for the infinite and finite horizon planning. Finally, we apply these findings in  Section 4  to  some examples in the genuinely non dominated setup.

\section{The   Merton problem under  ambiguity aversion}
Consider the problem of an agent investing in $n$ risky assets
and a riskless asset.  Specifically, we work under the Black-Scholes-Merton market model assumptions. Namely,   the riskless rate $r$ is  constant and the risky  assets dynamics are, for each  $i=1,\ldots,d$:
\begin{equation}
d S_t^i = S_t^i \left(\mu^i d t +\sum_{j=1}^d \sigma^{ij} d W_t^j\right )
\end{equation}
where
$\sigma^{ij}$ and $\mu^i$ are constants and $W$ is a standard, $d$-dimensional Brownian motion on a filtered space $(\Omega, (\mathcal{F}_{t})_{t\geq 0}, P)$. Assume that $d\geq n$, so the market is allowed to be incomplete.   In matrix-vector form,
the above equation becomes:
\begin{equation*}
dS_t = Diag(S_t)( \mu dt + \sigma  dW_t)
\end{equation*}
where by $Diag(S_t)$ we denote the diagonal $n\times n$  matrix with $i$-th diagonal element equal to $S^i_t$. In addition,  $\sigma$ is required to have full rank,  so that the variance-covariance matrix $\Sigma = \sigma\sigma' $ is  invertible. Here and in what follows,  $'$ denotes the transpose operation. \\
 \indent Given the initial endowment $w_0$,  the investor is allowed to trade and consume in a self-financing way. To be explicit, let   $h=(h_t)_t$ denote  the $n$-dimensional progressively measurable process, representing the number of shares of each asset held in portfolio,  and     let the progressively measurable, nonnegative, scalar process  $c$ indicate  the consumption stream.  Assume  also that $\int_0^\cdot c_s ds$ is finite $P-$a.s.   Then, the wealth process is governed by the following stochastic differential equation:
$$
dw_t =  (r w_t + h'_t  Diag(S_t)  (\mu-r {\bf 1})   -c_t) dt + h'_t Diag(S_t) \sigma  dW_t
$$
in which ${\bf 1}$ is the $d$-vector with all components equal to one.  It is convenient to recast the wealth equation  by the vector process $\theta$ of cash value allocated in each risky asset, i.e. $\theta_t := Diag(S_t)h_t$. Thus,
\begin{equation}\label{wealth1}
dw_t =  (r w_t +  \theta'_t  (\mu-r {\bf 1})     -c_t) dt + \theta' \sigma  dW_t.
\end{equation}
The pair $(\theta_t,c_t)$ is admissible for the initial
wealth $w_0$ if the wealth process $w_t$ given by (\ref{wealth1}) remains
 $P$-a.s.  non-negative at all times.
Let ${\cal A}^{P}(w_0)$ be the set of all admissible $(\theta,c)$ pairs for initial wealth $w_0$. Note that the admissible set depends only on the equivalence class of $P$.
 The agent is then trying to choose $(\theta,c) \in {\cal A}^P(w_0)$,  so as to maximize the expected utility   from  running consumption and terminal wealth:
 $$
\sup_{(\theta, c) \in {\cal A}^P(w_0) }\mathbb{E}[\int_0^Tu(t,c_t)dt + u(T,w_T)].
$$
The utility function $u$ is assumed to be concave, increasing in the second argument and measurable in the first. This class of stochastic control problems is known under the name of Merton problem. It includes a number of specific cases, among which the infinite horizon planning. In fact,   if  $ T= \infty$, $w_\infty: = \limsup_{t\rightarrow \infty} w_t$ and $ u(\infty, \cdot) = 0$ the above optimization problem becomes:
 $$ \sup_{(\theta, c) \in {\cal A}^P(w_0) }\mathbb{E}[\int_0^\infty u(t,c_t)dt]. $$
So far, the exposition is classical, and can be found in many textbooks.  The reader is referred to
the new   \cite[Chapter 1]{rogers},
 for a remarkably didactic approach.  \\
 \indent However, things change quite a bit if the agent is diffident about
the (constant)    estimates $\hat \mu $  and full rank matrix $\hsig$, for the drift and  volatility matrix  of  the risky assets respectively.   Assume from now on that $\Omega$ is the Wiener space of continuous functions, with the natural filtration $\mathbf{F}=(\mathcal{F}_t)_{t\geq 0}$.   Our investor
  assumes that the `true' volatility $\sigma$ is a progressively measurable matrix, and  such that the  variance-covariance   $\Sigma =\sigma \sigma'$  takes  values in  some fixed compact set $K $ of $n\times n$ invertible matrices,  containing $\hat \Sigma$:
  $$ {S}:= \{ \sigma \in \mathbb{R}^{n\times d} \mid   \Sigma \in K \}.  $$
  Let us denote by $\mathcal{S} =\{ \sigma \text{ progr mis}\mid \sigma_t(\omega)  \in S  \text{ for all } \omega, t \} $. This choice is in line  with empirical practice, as $\Sigma$ is the estimated object, not the volatility $\sigma$.
The uncertain drift  is also assumed to be progressively measurable,   and for a given realization of $\sigma$  it is  allowed to vary in
$$
U_\epsilon(\sigma) = \{ u \in \mathbb{R}^n\mid  {(  u  - \hat \mu)' { \Sigma ^{-1} }(  u  - \hat \mu)}\leq \epsilon^2 \text{ for all } \omega, t \},
$$
that is, in an   ellipsoid centered at $\hat \mu$   with radius  $\epsilon$. Let us denote the set of plausible drifts and volatilities  by
 $$ \Upsilon: =\{ ( \mu, \sigma) \text{ progr. meas. } \mid  \sigma \in \mathcal{S},    \mu_t(\omega)  \in U_\epsilon(\sigma_t(\omega))  \}$$
 and let $ \Upsilon_{\sigma}$ denote its $\sigma$-section.  Different choices of $(\mu, \sigma) \in \Upsilon$  correspond to considering  different  probabilities  on the Wiener space, namely those under which the risky assets evolve with the prescribed coefficients.  These probabilities are orthogonal to each other across the sections  $\Upsilon_{\sigma}$.\\
\indent For a fixed choice of  the   process  $\sigma$  however,  the Girsanov theorem  ensures that  all the    vector processes $\mu \in  \Upsilon_{\sigma}$  correspond to  probabilities that are equivalent to each other     on $\mathcal{F}_t$ for all $t>0$.   We would like to describe these equivalent changes of measure  as a \emph{function} of $\mu$. To this end,    let us select a reference  probability  corresponding to $(\hmu , \sigma)$  and call it $P^{\hmu, \sigma}$.   Market   incompleteness implies that  the probability in $\Upsilon_\sigma$ under which the risky assets evolve with drift $\mu$  is not unique. However,  such measures  can be fully  parametrized as  probability  changes  with respect to  $P^{\hmu,\sigma}$.
 And a minimal choice    (see also Remark \ref{nonconvex} below) is selecting for each $\mu$ the probability $P^{\mu, \sigma}$ corresponding to the measure change  given by the
  Dol\'{e}ans exponential $\mathcal{E}(\int_0^\cdot \varphi^\mu dW) $, where $$ \varphi_t^\mu:= \sigma_t'  \Sigma_t ^{-1}(\mu_t -\hat \mu).    $$
Such  selection does not reduce generality, since what matters in our context of expected utility maximization are only the distributional properties of the risky assets.  Therefore, we have now a one-to-one correspondence between  elements of $\Upsilon$ and probabilities on $(\Omega, \mathbf{F})$, namely   the possible prior models are  given by
 $$ \mathcal{P} = \{  P^{\mu, \sigma}  \mid (\mu, \sigma)\in \Upsilon \}.  $$
Now, the wealth process evolves under each  ${P}^{\mu, \sigma} $ according to
\[
dw_t = (r w_t +\theta_t'(\mu_t-r {\bf 1}) - c_t) dt + \theta'_t  \sigma_t  d{W}
\]
where $ {W}$ is a ${P}^{\mu,\sigma}$-standard Brownian motion.  Finally,  let us  call the investment/consumption pair $(\theta , c )$   \emph{(robust)  admissible}
for the initial positive wealth $w_0$ if in addition to the measurability and integrability assumptions already made at the beginning of this section,  the wealth process remains non-negative for all  $P \in \mathcal{P}$:
 $$\Arob : = \cap_{P \in \mathcal{P}}{ \cal{A}}^P(w_0) =  \cap_{\sigma\in\mathcal{S} } {\cal{A}}^{P^{\hmu, \sigma}}(w_0). $$
The equality on the rhs holds because, for a given $\sigma$, the   admissible class is invariant for different choices of $\mu \in \Upsilon_\sigma$.
 The ambiguity averse  investor takes a prudential worst case approach, and    faces the following
  robust  Merton  problem:
 \begin{equation}\label{robmer}
u_{opt}(w_0): =\sup_{(\theta,c)\in  \Arob} \inf_{(\mu, \sigma)\in \Upsilon}  \mathbb{E}^{(\mu,\sigma)} \left [\int_0^T u(t,c_t)dt + u(T, w_T) \right ].
 \end{equation}
 It is clear that  more conservative portfolio choices are made when the  uncertainty set $\Upsilon$
 is larger, while an ambiguity-neutral investor   sets  $\mathcal{S} =\{\hsig\}$ and  $\epsilon$ equal to zero, thus facing  a classical Merton problem for  the model $P^{\hmu, \hsig}$.
\begin{rem}\label{nonconvex}
 Fix $\sigma \in \mathcal{S}$. Any change of drift, say from $\hat \mu$ to $\mu$,  corresponds to a change of measure from $P^{\hmu, \hsig}$ to a probability $\widetilde{P}$ with density process $ \frac{d\widetilde{P}}{dP}$ given by a  Dol\'{e}ans exponential $ \mathcal{E}(\int_0^\cdot \varphi dW)$.  The suitable $\varphi$s can be characterized as those in the form:
 $$ \varphi_t =  \varphi_t^\mu + \psi_t  $$
 in which $\psi$ is a (sufficiently integrable) progressively measurable process, belonging  $dt \otimes dP^{\hmu, \sigma}$-a.e. to $\mathrm{ker}(\sigma_t(\omega))$.
   Chen and Epstein \cite{chenep} call the process $\varphi $  the \emph{market price of ambiguity.}   Denote this class of probabilities by $\mathcal{P}^{\sigma}$.  Elementary optimization shows  $\varphi^\mu$  is minimal, in the sense that it has the smallest pointwise $ \mathbb{R}^n$-norm among $\mathcal{P}^{\sigma}$:
   $$  \|\varphi^\mu_t(\omega)\|^2 =   \min_{\widetilde{P} \in \mathcal{P}^\sigma} \|\varphi_t(\omega)\|^2  \ \ \  dt\otimes dP^{\hmu, \sigma}- a.e. $$
Thus, the ellipsoidal ambiguity on the drift in $\Upsilon_\sigma$ can be recast into a (non convex!) condition on  the market price of uncertainty $\varphi$, namely
$$ \mu \in \Upsilon_\sigma  \ \text{ iff } \  \min_{\widetilde{P} \in \mathcal{P}^\sigma} \|\varphi_t(\omega)\|^2 \leq \epsilon^2   \ \ dt\otimes dP^{\hmu, \sigma}- a.e.  $$
This should be contrasted with the ubiquitous  conditions in the literature (see e.g., \cite{hs}, \cite{chenep}) which require the market price of uncertainty to be valued in a convex set.
\end{rem}
The resolution of \eqref{robmer} is based on the next robust version of the  verification theorem.  Although formulated for the sets $\mathcal{P}, \Upsilon$, its validity is general and does not rely on any  specific parametrization of the set of prior models.
\begin{thm}\label{verif}
For a shorthand, call $\nu=(\mu, \sigma)$ a generic element of $\Upsilon$.  Suppose that:
\begin{enumerate}
  \item there exists a function $V:[0,T] \times \mathbb{R}^+ \rightarrow \mathbb{R}$,  which is continuous on  $[0,T]\times \mathbb{R}^+$ and   $C^{1,2}$ on $[0,T)\times \mathbb{R}^+ $,  verifying  $V(T,.)=u(T,.)$;
  \item for any $(\theta, c)$ there exists an optimal solution $\nu(\theta, c) \in \Upsilon$ of the inner minimization in \eqref{robmer}, such that
  \begin{equation}\label{y}
Y_t = Y_t^{(\theta, c)} \equiv V(t,w_t) + \int_0^t u(s,c_s)ds
\end{equation}
is a $P^{\nu(\theta, c)}$-supermartingale;
  \item  there exist some $(\bar \theta, \bar c) \in \Arob $ such that the corresponding $ \overline{Y} $ is a $P^{\nu(\bar \theta, \bar c)} $- martingale.
\end{enumerate}
Then $(\bar \theta , \bar c)$ is optimal for the problem \eqref{robmer} and   $V(0,w_0)$ is the optimal value function, namely $\uopt(w_0) =V(0,w_0)$.
\end{thm}
\begin{proof}
The proof is a simple modification of the
   Davis-Varaiya Martingale Principle of Optimal Control   \cite[Theorem 1.1]{rogers}. In fact, by the supermartingale property of $Y$  under $P^{\nu(\theta, c) }$, and by $V(T,.)=u(T,.)$,  we have:
  $$ \mathbb{E}^{\nu(\theta, c)}[Y_T] = \mathbb{E}^{\nu(\theta, c)}[  \int_0^T u(s,c_s)ds + u(T,w_T)] \leq Y_0 =  V(0,w_0).   $$
  Taking the supremum over $ \Arob$ gives $   \uopt (w_0) = \sup_{\Arob} \mathbb{E}^{\nu(\theta, c)}[  \int_0^T u(s,c_s)ds + u(T,w_T)]\leq V(0,w_0)$.
  Since by assumption for  some $\overline{\nu} = (\bar\theta , \bar c )$ the process $\overline{ Y}$ is a martingale under $P^{\nu(\bar\theta , \bar c)} $, then  $ \mathbb{E}^{\overline{\nu}}[\overline{Y}_T] = Y_0 = V(0,w_0) $ and the conclusions immediately follow.
 \end{proof}

Now, the usage of the  verification theorem to solve the ambiguity-averse investor's problem  is quite intuitive.  Given a specific utility function, one looks for a function $V$ satisfying the
premises of the theorem. Using It\={o}'s formula,   any process  $Y$ as in \eqref{y} verifies under $P^{\nu}$ the following SDE:
 \begin{equation}\label{ydiff}
   d Y_t = \{ u(t,c_t)+ V_t + V_w(r w_t + \theta_t'(\mu_t-r{\bf 1}) - c_t)+ \frac{1}{2}   {\theta_t'\Sigma_t \theta_t} V_{ww}\} dt +  V_w \theta_t' \sigma_t dW.
 \end{equation}
 To make   $Y$ a supermartingale under every $P^{\nu(\theta, c)}$, and a martingale for some $P^{\nu(\theta^*, c^*)}$,
the maximum over $(\theta,c) \in \mathbb{R}^n \times \mathbb{R}^+$ of  the minimum of  $\nu  \in \Upsilon$  must be equated to zero.  At this point, some other specific structure on $\sigma$ must be assumed in the fully  incomplete market case, like e.g., dependence on a correlated nontraded asset as in \cite{hs}.  \\
 \indent In the rest of the paper however,  we focus on the complete market case, i.e. we assume $\sigma$ is a square matrix.   Then,   a max-min nonlinear PDE arises from \eqref{ydiff}:
$$
  \max_{(\theta,c) \in \mathbb{R}^n \times \mathbb{R}^+} \, \min_{ \sigma \in {S}, \mu \in U_\epsilon(\sigma) } \left[u(t,c)+ V_t + V_w(r w + \theta'(\mu  -r{\bf 1}) - c)+ \frac{1}{2} {\theta'\Sigma   \theta}  V_{ww}
\right] =0,$$
which is of the  Hamilton-Jacobi-Bellman-Isaacs type. In the following, we simply refer to it as   the  robust HJB equation.
The minimization of the drifts for  $\sigma$ fixed gives:
$$\min_{\mu \in \mathbb{R}^n } \{\theta' \mu: (\mu-\hat \mu)'\Sigma^{-1}(\mu - \hat \mu) \leq \epsilon^2\},$$
which  is a simple exercise in Karush-Kuhn-Tucker necessary and sufficient conditions.  When $\theta \neq 0$, the optimal solution is
$$\mu(\theta) = \hat \mu - \epsilon \frac{ \Sigma \theta}{\sqrt{\theta'\Sigma\theta} }.$$
 Substituting it back in the  robust HJB, we get
\begin{equation}\label{hjb0}
  \max_{(\theta,c) \in \mathbb{R}^n \times \mathbb{R}^+}  \min_{\sigma \in  S} \left[u(t,c)+ V_t + V_w(r w + \theta'(\hat \mu-r{\bf 1}) - \epsilon \sqrt{\theta'\Sigma\theta} - c)+ \frac{1}{2} {\theta'\Sigma\theta}  V_{ww} \right]=0,
\end{equation}
which covers also the case $\theta=0$.    Since the value function $V$ will be increasing (and concave) in the wealth $w$,  we are minimizing  a concave function over  ${S}$:
\begin{equation}\label{inner-sigma}
  \min_{  \sigma \in  {S}}   \left( - V_w   \epsilon \sqrt{\theta'\Sigma\theta}  + \frac{  V_{ww} }{2} {\theta'\Sigma\theta}  \right ).
  \end{equation}
Notice  that the function to be optimized depends  on $\sigma$ only via  the quadratic form $\theta'\Sigma\theta $, and that its derivative wrt $ y:=\theta'\Sigma\theta $ is positive as $V_w>0, V_{ww}<0$. Therefore,  the optimizers are not unique in general and compactness of $S$ is of essence.     The actual computations depend on the specification of $S$, and may be    quite easy   when the set    $S$ is defined precisely in terms of   constraints  on the quadratic form, as we show in Section \ref{ambsig}.\\
 \indent Let us denote by $\bsig(\theta)$  an optimizer, where $\bsig(\theta)$ is the Cholesky factorization of an optimal varcov matrix $\overline{\Sigma}(\theta)$.   Since  the optimal value of quadratic form $\theta'\overline{\Sigma}(\theta)\theta$  does not depend on the specific choice of $\bsig(\theta)$,     one    obtains  the  PDE
 \begin{equation}\label{hjb1}
  \max_{(\theta,c) \in \mathbb{R}^n \times \mathbb{R}^+}   \left[u(t,c)+ V_t + V_w(r w + \theta'(\hat \mu-r{\bf 1}) - \epsilon \sqrt{\theta' \overline{\Sigma}(\theta)\theta} - c)+ \frac{1}{2} {\theta'\overline{\Sigma}(\theta)\theta}  V_{ww} \right]=0,
\end{equation}

In the main applications we present in Section 3 and  4, $ \overline{\Sigma}(\theta)= \overline{\Sigma}$, namely  a constant.   When this is the case, the above  equation is equivalently viewed as stemming from the worst-case $( \bmu, \bsig)$, in which $\bsig$ is the Cholesky factorization of $\overline{\Sigma}$,    $\bmu = \hat \mu - \epsilon \frac{ \overline{\Sigma} \theta}{\sqrt{\theta'\overline{\Sigma}\theta} }$, and using the worst couple in the wealth equation:
\begin{equation}\label{wcwealth}
d w_t =   (r w_t +\theta_t'(\hat{\mu}_t-r {\bf 1}) - \epsilon \sqrt{\theta_t' \overline{\Sigma} \theta_t}- c_t) dt + \theta'_t  \bsig  d{W}  .
\end{equation}
Therefore,   the  general problem  \eqref{robmer} becomes equivalent to a robust utility maximization when there is (ellipsoidal) uncertainty in the drift only.  \\
 \indent The techniques employed to solve \eqref{hjb1} are then standard, and rely on educated guesses at the form of the solution.  If the solution $V$ of the robust HJB  can be found explicitly, then it is the candidate to be the value function we are looking for. Finally, to conclude that $V$ is indeed the value function, one must check that it verifies items 1, 2 and 3 in Theorem \ref{verif}.

\section{The  robust  power utility problem with non ambiguous $\sigma$}
We assume in this Section  that the investor has a  power utility function, and there that is no uncertainty on the (constant) square volatility matrix, namely $\mathcal{S} = \{ \hsig \}$ .  Lack of uncertainty in the volatility  may be empirically justified by the consideration that mean returns are subject to imprecision to a much higher extent than volatilities.   We provide explicit solutions  both when the planning horizon is   finite and infinite. To avoid notation overload, and for the next usage in Section 4,  \emph{ we drop the hat over} $\sigma$.
Also,  we denote by  $$\drifts:=\{ \mu \text{ progr meas} \mid  \mu_t(\omega) \in U_\epsilon(\sigma) \text{ for all } \omega \}, $$
and by $\mathbb{E}^\mu$ the expectation under $P^{\mu, \sigma}$.
\subsection{The infinite horizon  planning}
\subsubsection{Resolution of  the robust HJB equation}
Let us assume the investor has   CRRA power utility from intertemporal consumption:
$$
u(t,x) = \mbox{e}^{-\rho t} \frac{x^{1-R}}{1-R},
$$
where $\rho$ and $R \neq 1$ are   positive constants, modeling the time impatience rate and relative risk aversion respectively.    In the infinite horizon case, we wish to find the solution of:
\begin{equation}\label{mer1-nosig}
u_{opt}(w_0, \epsilon) =    \sup_{(\theta,c) \in {\Arob} } \inf_{\mu \in \drifts} \mathbb{E}^{\mu} \left[
\int_0^\infty \mbox{e}^{-\rho s} \frac{c_s^{1-R}}{1-R} ds\right],
\end{equation}
when the problem is well-posed, i.e. when it has a finite value\footnote{When $\mathcal{S}=\{\sigma\}$,  we remark that  $\Arob$ coincides with the classic set of admissible plans $\mathcal{A}^{P^{\hmu, \sigma}}(w_0).$}.  \emph{Assume for the moment that  this is the case and also  that both the inner infimum (for a fixed $(\theta,c) \in \Arob$) and the outer  supremum are attained.}
The properties of the problem imply, exactly as in the classic case,
that a guess at the value function takes the form
$$
V(t,w) =  \gamma_{\epsilon}^{-R} u(t, w).
$$
The positive  constant $\gamma_\epsilon$ has to be determined, and we use $\epsilon$ as subscript to
highlight the dependence on the radius of ambiguity $\epsilon$.
  With this guess, let us solve (\ref{hjb1}). The optimization over
  $c$ trivially
results in
$$
\bar{c} = \gamma_\epsilon w,
$$
with
$$
\max_{c} \{u(t,c) - c V_w\} =  {e}^{-\rho t}\frac{R}{1-R}(\gamma_\epsilon w)^{1-R}.
$$
The residual optimization  is
$$\max_{\theta} \left[ {e}^{-\rho t}\frac{R}{1-R}(\gamma_\epsilon w)^{1-R} + V_t +  V_w(r w + \theta'(\hat \mu-r{\bf 1}) - \epsilon \sqrt{\theta'\Sigma\theta}) + \frac{1}{2} {\theta'\Sigma\theta}  V_{ww} \right ].$$

The function to be maximized is concave in $\theta$, and smooth in   $ \mathbb{R}^n\setminus \{0\}$. The first order conditions are thus necessary and sufficient for optimality  in  $\theta \neq 0$. So, by equating the gradient to zero we obtain:
$$
\theta (s)   = \frac{- s V_w}{ s V_{ww} - V_w \epsilon}  \Sigma ^{-1}
(\hat \mu - r {\bf 1}),
$$
where  $s := \sqrt{\theta' \Sigma  \theta}$.  We are left with  $$  s^ 2  = \theta(s)' \Sigma  \theta(s) $$
Set $$H: = \sqrt{(\hat \mu - r {\bf 1})'\Sigma ^{-1} (\hat \mu - r {\bf 1})}, \ \ \ H_\epsilon: = \Heps.$$
  The above equation has  a   positive root, given by:
$$
\bar{s}=-\frac{V_w  H_\epsilon }{V_{ww}}
$$
if and only if   $ H_\epsilon > 0$.    If $  H_\epsilon \leq 0$, the  optimal solution is necessarily $\bar{\theta} =0$.    Finally, if  $H_\epsilon^+$ denotes the positive part of $H_\epsilon$, the following is a compact way of writing  the optimal solution in both cases:
$$
\bar{\theta} =  w \frac{ H^+_\epsilon }{R H}  \Sigma^{-1} (\hat \mu - r {\bf 1})  $$
Now,  $\gamma_\epsilon$ is found by substituting these $\bar{c}$ and $\bar{\theta} $ back into (\ref{hjb1}) and solving for the constant. Straightforward calculations result in:
\begin{equation}\label{gamma}
  \gamma_\epsilon =  \frac{\rho + (R-1)( r   + \frac{1}{2}\frac{(H_\epsilon^+)^2}{R})}{R}
\end{equation}
which for $\epsilon=0$ falls back to the  constant $\gamma_0=  \frac{\rho + (R-1)[ r +   \frac{1}{2}\frac{ H^2}{R}] }{R} $ of the classic case.
Therefore, the value function $V$ of the problem is found as
$$
V(t,w) = \gamma_\epsilon^{-R} u(t,w).
$$
This of course holds \emph{as long as} $\gamma_\epsilon >0$, which is shown below to be a  necessary and sufficient condition for the well posedness of the robust Merton problem.
\subsubsection{The verification and comparison with the classic Merton problem}
\begin{prop} \label{prop1}
 The infinite-horizon robust Merton problem under ellipsoidal ambiguity of mean returns:
$$
u_{opt}(w_0, \epsilon) =    \sup_{(\theta,c) \in {\Arob} } \inf_{\mu \in \drifts} \mathbb{E}^{\mu} \left[
\int_0^\infty \mbox{e}^{-\rho s} \frac{c_s^{1-R}}{1-R} ds\right],
$$
is well posed if and only if $\gamma_\epsilon  $ in \eqref{gamma} is strictly positive. In this case,  the optimal value  is
$$ u_{opt}(w_0, \epsilon)  = V(0, w_0)= \gamma_\epsilon^{-R}   \frac{(w_0)^{1-R}}{1-R}, $$
and the  optimal controls are:
$$
\bar{\theta}_t=    \bar{w}_t \pi_\epsilon   \, ,  \  \bar{c}_{t} = \gamma_\epsilon \bar{w}_t,
$$
with optimal portfolio proportions  vector   given by
$$\pi_\epsilon:= \frac{H^+_\epsilon  }{R H}  \Sigma ^{-1} (\hat \mu - r {\bf 1})$$
 The worst case drift is constant:  $$\bar{\mu}:=\mu(\bar{\theta})= \hat \mu -\epsilon \frac{\Sigma}{\sqrt{\pi'_\epsilon \Sigma\pi_\epsilon}} \pi_\epsilon, $$
and the optimal wealth process has $P^{\bar{\mu},\sigma}$ dynamics given by
\begin{equation}\label{optwealth}
\bar{w}_t = w_0 \exp{\left (\pi_\epsilon \sigma    W_t + (r + \frac{(H_\epsilon^+)^2(2R-1)}{2R^2}-\gamma_\epsilon)t\right)}.
\end{equation}
\end{prop}
\begin{proof}
The proof is split into two steps.
\begin{enumerate}

  \item  If $\gamma_\epsilon \leq 0$  then  $u_{opt}(w_0, \epsilon) = \infty$. Note first that this case can only happen when $0<R<1$. The proof here closely follows the lines of  \cite[Section 1.6]{rogers}.
            \begin{itemize}
            \item[1-a)] Assume $\gamma_\epsilon <0$. Then, consider a couple of controls that are both proportional to the wealth:
          $$ \theta_t = w_t \pi , \ \ c_t = \lambda w_t \text{ with }  \lambda >0.$$
          If we substitute them into (7), then  the solution is the positive wealth
          $$ \widetilde{w}_t= w_0 \exp{ \left( (r +\pi'(\hat \mu -r {\bf 1} )- \epsilon \sqrt{\pi' \Sigma \pi}- \lambda - \frac{1}{2}\pi' \Sigma \pi )t + \pi'\sigma W_t   \right) } $$
           so that:
          $$ u_{opt}(w_0, \epsilon)  \geq \mathbb{E}^{\mu(\theta)}[\int_0^\infty e^{-\rho t} \frac{\lambda^{1-R}}{1-R} (\widetilde{w}_t)^{1-R} dt]. $$
          An application of the stochastic Fubini's theorem  shows that the latter is proportional to
         $$ \int_0^\infty \exp{\left[ t \left( -\rho + (1-R)(r +\pi'(\hat \mu -r {\bf 1} )- \epsilon \sqrt{\pi' \Sigma \pi}- \lambda - \frac{R}{2}\pi' \Sigma \pi)    \right )   \right]} dt.$$
         If there exist some $(\pi, \lambda)$ for which the exponent is positive, then the integral diverges and  the value function is infinite. For fixed $\lambda$, the maximum over $\pi$ in the exponent is attained for $\bar{\pi} =  {\pi_\epsilon} $, and the value is
         $$-\rho + (1-R) (r + \frac{(H_\epsilon^+)^2}{2R} -\lambda)  = - R \gamma_\epsilon - \lambda (1-R),$$
         which is positive for $\lambda $ small enough.
    \item[1-b)] If $\gamma_\epsilon =0$, take $\theta_t = w_t  \pi_\epsilon$, and $c_t = \frac{k}{1+t} w_t$ for some constant  $k>0$.   This choice leads to
     $$   u_{opt}(w_0, \epsilon)  \geq \int_0^\infty e^{-\rho t} \frac{1}{1-R } \frac{k^{1-R}}{(1+t)^{1-R} }   E^{\mu(\theta)}[(\widetilde{w}_t)^{1-R}] dt.  $$
   Now, $\mathbb{E}^{\mu(\theta)}[ (\widetilde{w}_t)^{1-R} e^{-\rho t}  ] = e^{   -  (1-R)\int_0^t \frac{k}{1+s}ds} = e^{ -(1-R)k \ln (1+t)}= \frac{1}{(1+t)^{k(1-R)}}$ when $t\rightarrow \infty$. Therefore, the   integrand is asymptotic to $ (1+t)^{-(k+1)(1-R)}$  and hence  the integral diverges if e.g., $ k = \frac{1}{1-R} -1 $.
  \end{itemize}
 \item  If $\gamma_\epsilon > 0 $, then the optimal processes/value are as given in the statement of the proposition.\\
       Substituting the candidate optimal controls $(\bar{\theta}_t,\bar{c}_t)$ into
(\ref{wcwealth}),  and solving for the candidate optimal wealth one gets \eqref{optwealth}, which can be further simplified to:

$$
\bar{w}_t = w_0 \exp{\left ( \pi_\epsilon   W_t + \frac{1}{R}( r-\rho - \frac{(H_\epsilon^+)^2}{2})t \right)}
$$
The process $\bar{w}$ is a deterministic scaling of  Geometric Brownian motion, as well as the process   $  \bar{v}_t: = (\bar{w}_t)^{1-R}$.  Now,
$$(1-R) \bar{Y}_t := (1-R) [ V(t, \bar{w}_t) + \int_0^t u(s, \bar{c}_s) ds] = \gamma_\epsilon^{-R}e^{-\rho t} \bar{v}_t   + \int_0^t e^{-\rho s} (\gamma_\epsilon)^{1-R} \bar{v}_s ds.   $$
has integrable maximal functional in every compact $[0,T]$ and by construction has zero drift term.  Henceforth, $\bar{Y}$  is a $P^{\bar{\mu},\sigma }$ martingale. \\
For what concerns other admissible controls $(\theta_t, c_t)$,   the process  $$ Y_t =   V(t, w_t) + \int_0^t u(s, c_s) ds $$
 with $w$ as in \eqref{wcwealth}, is  by construction a diffusion, which has the same sign as $(1-R)$, and which has  non positive drift under $P^{\mu(\theta),\sigma}$.
 \begin{itemize}
   \item  If $0<R<1$, then  any such $Y$ is positive.  By writing \eqref{ydiff} under $P^{\mu(\theta),\sigma}$, it is immediate to realize that  $Y$  is a  positive, decreasing scaling of a positive local martingale, namely  the Dol\'{e}ans exponential of  the process $X$ defined by:
        $$X:= \int_0^\cdot \gamma_\epsilon^{-R} \frac{e^{-\rho s} }{Y_s} w_s^{-R} \theta'_s\sigma dW. $$
        Henceforth,  $Y$ is  a  $P^{\mu(\theta),\sigma}$-supermartingale. In addition, $V(\infty, \cdot)  =u(\infty, \cdot) =0 $ so that the conditions of the Verification Theorem \ref{verif} are satisfied, and the proof in this case is complete.
   \item If $ R>1$, any $Y$ is negative. A simple modification of the   argument just used in the $0<R<1$ case only shows that $Y$ is  a \emph{local} supermartingale.   Therefore in this case  we show    the optimality of  $((\bar{\theta},\bar{c}), \bar{\mu} )$ in another way.  To this end, note that  the martingale property of $\bar{Y}$ as above, together with standard minimax inequalities,  gives
       \begin{eqnarray*}
          \mathbb{E}^{\bar{\mu}}[\int_0^\infty  e^{-\rho t} \frac{(\bar{c}_t)^{1-R}}{1-R} dt ]= \gamma_{\epsilon}\frac{ w_0^{1-R}}{1-R} \leq  u_{opt}(w_0, \epsilon)  \leq  \\
       \leq  \inf_{\mu \in \drifts} \sup_{\Arob} \mathbb{E}^{\mu}[  \int_0^\infty  e^{-\rho t} \frac{(c_t)^{1-R}}{1-R} dt ] \leq  \inf_{\mu \in U_\epsilon} \sup_{\Arob}\mathbb{E}^\mu[ \int_0^\infty  e^{-\rho t} \frac{(c_t)^{1-R}}{1-R} dt],
       \end{eqnarray*}
       so if we prove  that the first value on the left  is equal to the value of the last problem on the RHS, we are done. This is quite an easy task. In fact, for  a fixed constant $\mu \in U_\epsilon$ the inner supremum is a standard Merton problem. Hence,
       $$  \sup_{(\theta, c) \in \Arob} \mathbb{E}^\mu \left [ \int_0^\infty e^{-\rho t} \frac{(c_t)^{1-R}}{1-R}dt\right ] =  ( \gamma(\mu))^{-R} \frac{(w_0)^{1-R}}{1-R}    $$
       in which we pose $\gamma(\mu): =   \frac{\rho + (R-1)( r   + \frac{1}{2}\frac{(H (\mu))^2}{R})}{R}, $ with $H(\mu):= \sqrt{(\mu- r {\bf 1})' \Sigma ^{-1}(\mu - r {\bf 1})} $.  The residual minimization:
       $$ \inf_{\mu \in U_\epsilon} ( \gamma(\mu))^{-R} \frac{(w_0)^{1-R}}{1-R}   $$
       is then  a simple exercise, the minimizer being $\bar{\mu} $, so that $\gamma(\mu) = \gamma(\bar{\mu}) = \gamma_\epsilon$, which ends the proof.
 \end{itemize}
\end{enumerate}
\end{proof}

 Let us remark that  the optimal portfolio $\bar{\theta}$ preserves the
form of the Merton's Mutual Fund theorem. In fact, the optimal portfolio consists of an allocation between two fixed mutual funds, namely the riskless asset and the fund of risky assets given by $\Sigma^{-1}
(\hat \mu - r {\bf 1})$.  At each time point the optimal relative allocation of wealth is now dependent on the ambiguity aversion of
the investor in addition to his/her risk aversion through the coefficient:
$$
\frac{    H_\epsilon^+  }{R H}.
$$
 The above allocation naturally
collapses to the Merton allocation $ \frac{1}{R}$ for $\epsilon=0$. In case the radius of ambiguity
$\epsilon$ is greater than or equal to the market Sharpe ratio $H$, the optimal
control policy is not to invest at all into the risky assets.
 Since $ \frac{    H_\epsilon^+  }{R H} \leq \frac{1}{R} $,  the robust  Merton portfolio $\pi_\epsilon$
  has  smaller  positions in absolute value with respect to the classical Merton portfolio. To wit, both long and short positions are shrunk
with respect to the ambiguity-neutral portfolio. As expected, and  already anticipated in the Introduction,  robustness in the decisions lowers the optimal demand on equity, and thus offer a theoretical basis for a possible explanation of the equity premium  puzzle. \\
\indent  The consumption in the ambiguity averse case may be increased or curtailed, depending on the sign of $R-1$. In fact, when  the classical problem and its  ambiguity averse counterpart with $\epsilon >0$ are both well posed,   if $0<R<1$ then $\gamma_\epsilon > \gamma_0>0$, while if $R>1$  the opposite inequality chain holds.

\subsection{The finite horizon planning for non ambiguous $\sigma$}
Now the investor has a   CRRA power utility both from intertemporal and terminal consumption at time $T<\infty$:
$$
u(t,w) = \mbox{e}^{-\rho t} \frac{w^{1-R}}{1-R} \text{ for } 0\leq t< T \, \text{ and } u(T,w)=A \frac{w^{1-R}}{1-R}
$$
in which $A$ is a fixed positive constant. Here, we set the deterministic scaling of the CRRA power utility identical to that of the infinite horizon case to better highlight the similarities, but everything stated below holds also if $e^{-\rho t}$ is replaced by an integrable, positive and deterministic function $h(t)$.  We then wish to find the solution of:
\begin{equation}\label{mer2-nosig}
u_{opt}(w_0, \epsilon) =    \sup_{(\theta,c) \in {\Arob} } \inf_{\mu \in \mathcal{U}_\epsilon} \mathbb{E}^{\mu} \left[
\int_0^T \mbox{e}^{-\rho s} \frac{c_s^{1-R}}{1-R} ds +  A \frac{w_T^{1-R}}{1-R}\right],
\end{equation}
 Using the scaling properties of the CRRA utility,  the guess to the
  value function is  of the  form  $V(t,w) = f(t) \frac{w^{1-R}}{1-R}$ for some positive, differentiable function satisfying $f(T)=A$.  The  HJB equation \eqref{hjb1} now looks like
$$
\max_{(\theta,c) \in \mathbb{R}^n \times \mathbb{R}^+} \!\! \left[e^{-\rho t}\frac{c^{1-R}}{1-R} \!\!+ \!\!f'(t)\frac{w^{1-R}}{1-R}\!\! +\!\! f(t) w^{-R}(r w \!+\! \theta'(\hat \mu \!-\!r{\bf 1}) \!-\! \epsilon \sqrt{\theta'\Sigma\theta} \!-\! c) \!-\!\frac{R}{2}f(t)w^{-R-1} {\theta'\Sigma\theta}   \right]\!\!=\!0.
$$
Proceeding exactly as in the previous section, one obtains
$$
\bar{c}(t, w)  = w \left(\frac{e^{-\rho t}}{f(t)}\right)^{1/R} \;  \bar{\theta} = w \pi_\epsilon.
$$
Substituting the above  back into the HJB equation results in   a first order  ODE for $f$:
$$ \left \{ \begin{array}{c }
     f'(t) + k_\epsilon f(t) + R e^{-\frac{\rho}{R}t}(f(t))^{1-\frac{1}{R}}=0  \\
     f(T)= A
   \end{array}\right.
 $$
with
$$k_\epsilon:= (1-R) \left (r+  \pi_\epsilon'(\hat\mu - r{\bf 1}) -\epsilon \sqrt{\pi_\epsilon' \Sigma \pi_\epsilon} - \frac{R}{2} \pi'_\epsilon \Sigma \pi_\epsilon \right ) = (1-R) (r + \frac{(H_\epsilon^+)^2}{2R}).$$
With the substitution $f(t) = g(t)^R$, the ODE can be linearized and easily solved:
$$
g(t) =  A^{\frac{1}{R}} \exp{\left(\frac{k_\epsilon}{R}(T-t)\right)} + e^{-\frac{k_\epsilon}{R}t} \int_t^T \exp{\left (\frac{k_\epsilon-\rho}{R}s \right )}ds.
$$

Comparing this to the solution of \cite[Section 2.1]{rogers} the only changes are: 1)
the constant $k_\epsilon$, in which  $H^2$   is replaced by $(H^+_\epsilon)^2$ and 2) the optimal portfolio allocation, which is identical to the robust allocation case of the previous section. Obviously for an ambiguity neutral investor with
$\epsilon = 0$ we fall back to the finite horizon solution of the  Merton problem.
 \\
 \indent Let us conclude by summing up the results just found, leaving the verification to the reader.
\begin{prop}
 The finite horizon robust Merton problem under ellipsoidal ambiguity of mean returns
 $$
u_{opt}(w_0, \epsilon) =    \sup_{(\theta,c) \in {\Arob} } \inf_{\mu \in \drifts} \mathbb{E}^{\mu} \left[
\int_0^T \mbox{e}^{-\rho s} \frac{c_s^{1-R}}{1-R} ds +  A \frac{w_T^{1-R}}{1-R}\right],
$$
is always well-posed, and  admits the optimal controls:
$$
\begin{array}{cc }
\bar{\theta}_{t }= &  \frac{ \bar{w}_t  H^+_\epsilon  }{R H}  \Sigma ^{-1} (\hat \mu - r {\bf 1}) = \bar{w}_t \pi_\epsilon\\
 \bar{c}_{t} = & \bar{w}_t\frac{e^{-\frac{\rho}{R} t}}{g(t)} \hfill
\end{array}
$$
where
$$
g(t) =  A^{\frac{1}{R}} \exp{\left(\frac{k_\epsilon}{R}(T-t)\right)} + e^{-\frac{k_\epsilon}{R}t} \int_t^T \exp{\left (\frac{k_\epsilon-\rho}{R}s \right )}ds ,
$$
and $k_\epsilon = (1-R) (r + \frac{(H_\epsilon^+)^2}{2R})$.  The optimal  $\bar{\mu} = \mu(\bar{\theta})= \hat \mu -\epsilon \frac{\Sigma}{\sqrt{\pi'_\epsilon \Sigma\pi_\epsilon}} \pi_\epsilon, $ and the optimal wealth process $\bar{w}$  has   dynamics  under $P^{\bar{\mu},\sigma}$ given by:
$$
\bar{w}_t =w_0 \exp{\left[  \left( r + \frac{(H_\epsilon^+)^2}{2R^2}(R-1)\right) t + \int_0^t \frac{e^{-\frac{\rho}{R}s}}{g(s)}ds  + \pi_\epsilon'\sigma W_t \right]}.
$$
\end{prop}
\section{Examples with ambiguous $\sigma$} \label{ambsig}
In all the following examples, the volatilities are square, full rank, matrices.

\begin{ex}[The uncorrelated case]
 Suppose that  estimated volatility matrix   $\hsig$ is  diagonal. To wit, the risky assets returns are   (instantaneously) uncorrelated.   Further, we suppose that the ambiguity does not affect  correlations, namely the   ambiguity set $S$  is that of diagonal  matrices, whose diagonal $\Sigma$    lies  in some product
 $$ [\underline{\sigma}^2_1, \overline{\sigma}^2_1] \times \ldots  [\underline{\sigma}^2_n, \overline{\sigma}^2_n],$$
  with $\inf_{i} \underline{\sigma}_i>0$  and $\underline{\sigma} _i \leq \hsig_i \leq \overline{\sigma}_i $.   This is exactly the case examined by Lin and Riedel \cite{riedel}, where the problem is treated via a G-Brownian motion technique. \\
  \indent For a fixed  ambiguity radius $\epsilon >0$  on the drift,   the  solution of the residual inner minimization over $\sigma$ in the max-min HJB  \eqref{hjb0} becomes a triviality with  this diagonal uncertainty specification.   The unique worst case volatility  is  constant and it is the  `highest' one,     $\bsig = Diag(\overline{\sigma}_1, \ldots,  \overline{\sigma}_n)$ and does not  depend on $ \theta$.  Therefore, the  general problem  \eqref{robmer} becomes equivalent to a robust utility maximization with   volatility   $\bsig$  and  ellipsoidal uncertainty on the drift only, with radius $\epsilon$.  To give an explicit example, in the power utility case one ends up in solving  \eqref{mer1-nosig} or  \eqref{mer2-nosig} with $\sigma = \bsig$.  It is clear then that the verifications are identical to the ones just seen in the previous Section.
The resulting optimal relative portfolio  is  also  constant:
  $$ \pi_\epsilon (\bsig) =  \frac{\overline{H}^+_\epsilon  }{R \overline{H} }  {(\overline{\Sigma}})^{-1} (\hat \mu - r {\bf 1})$$
    in which $ \overline{H} = \sqrt{(\hat \mu - r {\bf 1})'\overline{\Sigma} ^{-1} (\hat \mu - r {\bf 1})}, \overline{H}^+_\epsilon  = (\overline{H} -\epsilon)^+$.
\end{ex}

\begin{ex}[Upper bound  on the quadratic form $\Sigma$]
This example can be seen as relaxation of the previous one, in the sense that we do not impose constraints separately on  each of  the eigenvalues of $\Sigma$, nor we assume that $\Sigma$ is diagonal.  We simply restrict the quadratic form induced by $\Sigma$ not to exceed a given  threshold $\overline{\lambda}^2>0$ on the unit sphere, with $\overline{\lambda} \geq \hat \lambda_M$, the latter being the maximum eigenvalue of $\hat \sigma$. This amounts to imposing the same   bound on the maximum eigenvalue of $\Sigma$.
Precisely,   the volatility is assumed to be valued in
  $$ {S}:= \{ \sigma \in \mathbb{R}^{n\times n} \mid \,   0<  x' \Sigma  x \leq \overline{\lambda}^2 \|x\|^2  \text{ for all } x \in \mathbb{R}^n, x \neq 0 \}  $$
The    minimizers $\bsig$  of the inner minimization in \eqref{hjb0}  are volatilities $\sigma$  with  maximum eigenvalue equal to  $\overline{\lambda}  $, and such that $\theta$ is an eigenvector relative to $\overline{\lambda}$. Finally, the robust HJB boils down to the concave maximization:
   $$ \max_{(\theta, c)\in \mathbb{R}^n\times \mathbb{R}^+} \left [ u(t,c) +V_t +V_w (r w +\theta'(\hmu-r \mathbf{1}) -\epsilon \overline{\lambda}\| \theta\|-c) + \frac{1}{2}\overline{\lambda}^2 \|\theta\|^2 V_{ww}\right] =0 $$
 Therefore,  the ambiguous volatility problem  is observationally equivalent to a   Merton problem with    volatility matrix   equal to $\overline{\lambda} I $,  and drift uncertainty radius $ \epsilon   $. So, from  here one proceeds exactly  as in Section 3. The optimal  relative portfolio is thus
 $$ \overline{\pi}_{\epsilon} (\bsig) =  \frac{(\frac{1}{\overline{\lambda}} \|\hmu - r \mathbf{1} \| -  \epsilon )^+ }{ \frac{R}{\overline{\lambda}}\, \|\hmu - r \mathbf{1} \|}\frac{1}{\overline{\lambda}^2 } (\hmu - r \mathbf{1})=
 \frac{( \|\hmu - r \mathbf{1} \| -\overline{\lambda} \epsilon )^+ }{R\,  \|\hmu - r \mathbf{1} \|}\frac{1}{\overline{\lambda}^2 } (\hmu - r \mathbf{1}) $$
\end{ex}
\vspace{0.4cm}
\begin{rem} Another interesting case of ambiguity specification on
volatility is based on defining a ball around an estimate of the variance/covariance matrix. However, it does not
lead to a closed form portfolio rule.

  The ambiguity on volatility is
represented here as membership to the set $S = \{ \Sigma \succeq 0: \|\Sigma - \hat{\Sigma}\|_F \leq \delta\}$ for some estimate $\hat{\Sigma}$ of the variance/covariance
matrix\footnote{It is more convenient to work on $\Sigma$ rather than on $\sigma$; $\|Y\|_F$ denotes the Frobenius norm of a $m \times n$ matrix $Y$, and is equal to $\mbox{Tr}(Y' Y)$.}, and a positive parameter $\delta$.  The inner
optimization problem
$  \max_{  \Sigma \in  {S}} V_w \epsilon \sqrt{\theta'\Sigma\theta}     -\frac{  V_{ww} }{2} {\theta'\Sigma\theta} $
is equivalently posed as the matrix optimization problem
in the space of symmetric $n \times n$ matrices:
$$  \max_{  \Sigma \in  {S}}  V_w \epsilon \sqrt{\langle \theta \theta', \Sigma\rangle }   -\frac{  V_{ww} }{2} {\langle \theta \theta', \Sigma\rangle}   $$
where $\langle X,Y\rangle = \mbox{Tr}(X^T Y)$ is the trace product (inner product in the space of symmetric $n \times n$ matrices). Denoting 
 rank-one matrix $\theta \theta'$ (obtained from the dyadic product) as $\Theta$, and passing to variables $X = \Sigma - \hat{\Sigma}$ we have the equivalent quadratically constrained   optimization (convex) problem in matrix variable $X$:
$$  \max_{ X}  \beta(\epsilon) \sqrt{\langle \Theta, X \rangle + \langle \Theta, \hat{\Sigma} \rangle}  +    \alpha \langle \Theta, X \rangle  $$
subject to
$$
\langle X,X \rangle \leq \delta^2
$$
where we defined $\alpha = -\frac{  V_{ww} }{2}$ and $\beta(\epsilon) = V_w \epsilon$ for convenience  and we omitted momentarily the term $\alpha \langle \Theta, \hat{\Sigma} \rangle$.

From the first-order optimality conditions, after some straightforward algebra we obtain the following optimal
(worst-case) matrix
$$
\bar{\Sigma} = \hat{\Sigma}+ \frac{1}{2 \bar{\lambda}}(\alpha + \frac{\beta}{2\xi}) \Theta
$$
where $A = \langle \Theta, \hat{\Sigma} \rangle$ and $B=\|\Theta\|_F$,
$
\xi = \sqrt{A + \delta B}$,
$\bar{\lambda} = \frac{1}{4}\frac{(2 \alpha \sqrt{A + \delta B} + \beta)B}{\delta\sqrt{A + \delta B}}$.
Note that $\bar{\Sigma}$ is positive definite.
 Unfortunately, substitution of the worst-case matrix into the robust HJB equation (\ref{hjb0}) yields a fourth-order polynomial
which does not lead to a closed-form portfolio rule. Nonetheless, a numerical
procedure can be used to find the optimal portfolio rule.
\end{rem}

\end{document}